\documentclass[twocolumn,11pt]{IEEEtran}

\usepackage{basicreq}

\newcommand{\fzt}{\tilde{\vf}_0}
\newcommand{\fot}{\tilde{\vf}_1}
\newcommand{\fit}{\tilde{\vf}_i}

\newcommand{\lmax}{\ell_{\mathrm{max}}}
\newcommand{\vel}{\varepsilon_{\ell}}
\newcommand{\velm}{\varepsilon_{\ell-1}}
\newcommand{\kla}{k_{\ell}^{(1)}}
\newcommand{\klb}{k_{\ell}^{(2)}}

\newcommand{\kza}{k_{0}^{(1)}}
\newcommand{\kzb}{k_0^{(2)}}
\newcommand{\nld}{n_{\ell_d}}
\newcommand{\pe}[1]{P_e^{(#1)}}
\newcommand{\ka}{k^{(1)}}
\newcommand{\kb}{k^{(2)}}
\newcommand{\peld}{P_e^{(\mathrm{loc})}}
\newcommand{\wtx}{\widetilde{x}}
\newcommand{\wtX}{\widetilde{X}}

\newcommand{\bx}{\bar{x}}

\begin{document}

\title{Local Decoding in Distributed Compression
  \thanks{S.\ Vatedka is with the Department of Electrical Engineering, Indian Institute of Technology Hyderabad, India.  V.\ Chandar is with DE Shaw, New York, USA.  A.\ Tchamkerten is with the Department of Communications and Electronics, Telecom Paris, Institut Polytechnique de Paris, France.\\The work of S.\ Vatedka was supported by a seed grant from IIT Hyderabad and a Start-Up Research grant (SRG/2020/000910) from the Science and Engineering Research Board, India.\\
This work was presented in part at the 2022 IEEE International Symposium on Information Theory~\cite{vatedka2022_conferenceversion}.}
}
\author{
\IEEEauthorblockN{Shashank Vatedka},
\and
\IEEEauthorblockN{Venkat Chandar},
\and
\IEEEauthorblockN{Aslan Tchamkerten}
}
\maketitle
 
\begin{abstract}
It was recently shown that the lossless compression of a single source $X^n$ is achievable with a notion of strong locality; any $X_i$ can be decoded from a {\emph{constant}} number of compressed bits, with a vanishing in $n$ probability of error. By contrast, we show that for two separately encoded sources $(X^n,Y^n)$, lossless compression and strong locality is generally not possible. Specifically, we show that for the class of ``confusable'' sources, strong locality cannot be achieved whenever one of the sources is compressed below its entropy. Irrespective of $n$, for some index $i$ the probability of error of decoding $(X_i,Y_i)$ is lower bounded by $2^{-O(\rwc)}$, where $\rwc$ denotes the number of compressed bits accessed by the local decoder. Conversely, if the source is not confusable, strong locality is possible even if one of the sources is compressed below its entropy. Results extend to an arbitrary number of sources.
\end{abstract}

\section{Introduction}

The amount of data generated in many applications such as astronomy and genomics has highlighted the growing need for compression schemes that allow to interact and manipulate data directly in the compressed domain~\cite{pavlichin2013human,pavlichin2018quest,chen2014data,hashem2015rise,ball2012public,uk10k2015uk10k,schadt2010computational}. Indeed, traditional compression schemes such as Lempel-Ziv~\cite{ziv1977universal,ziv1978compression} are suboptimal in this regard since the recovery of even a single message symbol necessitates to decompress the entire dataset. Accordingly, this paper focuses on providing random access in the compressed domain, where short fragments of data can be recovered without accessing the entire compressed sequence. 

For the single source setup, \cite{mazumdar2015local,tatwawadi18isit_universalRA} showed that a strong notion of locality holds: for any rate above entropy there exists an encoder and a local decoder which probes a constant number $\rwc$ (independent of $n$) of compressed symbols, and yet achieves vanishing error probability as $n$ grows. Note that the concatenation scheme where the source is decomposed into $n/b$ consecutive blocks of some size $b$, each of which independently compressed at a desired rate $R>H(X)$, is not strongly local. Indeed, any $X_i$ is independent of all $n/b$ sub-block codewords, except one which reveals $b$ message symbols, and $X_i$ in particular. Hence, only ``weak'' locality holds in the sense that for the local decoder error probability to vanish, the number of probed symbols---here equal to the sub-block codeword length $b\cdot R$---must grow with $n$. 

In this paper we address the question whether strong locality extends to the Slepian-Wolf distributed compression of two sources $X^n$ and $Y^n$: given $(R_1,R_2)$ within the Slepian-Wolf rate region, is it possible to design a fix-length compressor and a local decompressor with $\rwc=O(1)$, and whose error probability is $o(1)$ as $n$ grows? 

Obviously, if each source is compressed above its entropy then strong locality holds simply by duplicating the results of \cite{mazumdar2015local,tatwawadi18isit_universalRA} separately for each of the sources. Note also that the concatenation scheme---wherein $(X^n,Y^n)$ is decomposed into consecutive sub-blocks of size $b$ each of which encoded via Slepian-Wolf coding---achieves weak locality at any $(R_1,R_2)$ within the Slepian-Wolf rate region. So the interesting question is: does strong locality hold when at least one of the sources is compressed below its entropy?

Our main result says that strong locality is generally impossible. More precisely, suppose $p_{XY}$ is ``confusable'' in the sense that, for every $x_1$ and $x_2$ in $\cal{X}$ there exists $y\in \cal{Y}$ such that $p_{XY}(x_1,y)>0$ and $p_{XY}(x_2,y)>0$. In this case, we show that if $R_1<H(X)$, the probability of wrongly decoding $(X_i,Y_i)$ is lower bounded by $2^{-\Theta(\rwc)}$, for some index $1\leq i\leq n$. Moreover, this conclusion holds even if the decoder tries to decode only $X_i$ with the full cooperation of the $Y$-transmitter that provides $Y^n$ uncompressed. Conversely, if $p_{XY}$ is not confusable, then strong locality is possible for some $R_1<H(X)$ and $R_2=H(Y)$. 

Hence, when the source is confusable, the concatenation scheme is order optimal in the tradeoff between local error probability and number of probes. However, a drawback of the concatenation scheme is that both the encoding and the decoding are tied to the sub-block length $b$ which governs the error probability of the local decoder. Even if both codewords are  entirely probed, that is $\rwc=n(R_1+R_2)$, the error probability remains the same as if $\rwc=b(R_1+R_2)$. Thus, to lower the error probability of the local decoder, the encoding procedure must be modified accordingly. We address this limitation through a hierarchical compression scheme whose local decoder achieves an error probability that decreases as $\rwc$ increases, without modifying the encoding. Specifically, for any $(R_1,R_2)$ within the Slepian-Wolf rate region, and for every $1>\eta>2^{-2^{O(\log n)}}$, the local decoder achieves $\peld\leq \eta$ with $\rwc=\mathrm{poly}(\log(1/\eta))$.  

\subsection{Literature on locally decodable compression}	

Local decoding has been studied extensively in the context of compressed data structures by the computer science community;  see, {\it{e.g.}},~\cite{patrascu2008succincter,dodis2010changing,munro2015compressed,raman2003succinct,chandar2009locally,chandar_thesis} and the references therein. Most of these results hold under the word-RAM model which assumes that operations (memory access, arithmetic operations) on $ w $-bit words take constant time. The word size $w$ is typically chosen to be $\Theta(\log n)$ bits, motivated in part by on-chip type of applications where data transfer happens through  a common memory bus for both data and addressing (hence $w=\Theta(\log n)$ bits), and partly by the fact that certain proof techniques work only when $w=\Omega(\log n)$. 

	In the word-RAM model, it is possible to compress any sequence to its empirical entropy and still be able to locally decode any message symbol in constant time~\cite{patrascu2008succincter,dodis2010changing}. Most approaches modify the Lempel-Ziv class of algorithms to provide efficient local decodability~\cite{kreft2010lz77,dutta2013simple,bille2011random}. Similar results also hold for compression of correlated data~\cite{viola2019howtostore}, and efficient recovery of short substrings of the message~\cite{sadakane2006squeezing,gonzalez2006statistical,ferragina2007simple,kreft2010lz77}. However, all of these schemes require the local decoder to probe at least $O(\log n)$ compressed bits to recover any source symbol.
		
        In this work, the decoding cost is measured by the \emph{number of compressed bits} that need to be accessed in order to recover a single source symbol, sometimes referred to as the \emph{local decodability}~\cite{makhdoumi_locally_2015}, or the \emph{bit-probe complexity} in the literature~\cite{nicholson2013survey}.
        
The problem of locally decodable source coding of random sequences was first studied by~\cite{makhdoumi2013locally-arxiv,makhdoumi_locally_2015}. These works showed that any compressor with $\rwc=2$ cannot achieve a rate below the trivial rate $\log|\cX|$, and any \emph{linear} source code that achieves $ \rwc=\Theta(1) $ necessarily operates at a trivial compression rate ($R=1$ for binary sources).
	Later,~\cite{mazumdar2015local} showed that for any $\epsilon>0$, rate $ H(X)+\varepsilon $ is achievable with local decodability $\rwc= \Theta(\frac{1}{\varepsilon}\log\frac{1}{\varepsilon}) $. Moreover, for non-dyadic sources, $ \rwc=\Omega(\log(1/\varepsilon)) $ for any compression scheme that achieves rate $ H(X)+\varepsilon $. Inspired by \cite{mazumdar2015local}, a compressor of Markov sources  was given in~\cite{tatwawadi18isit_universalRA} which achieves a rate-locality tradeoff $(R=H(X)+\varepsilon,\rwc= \Theta(\frac{1}{\varepsilon^2}\log\frac{1}{\varepsilon}) $. A common feature of the code construction in both papers is the use of the bitvector compressor of Buhrman \emph{et al.}~\cite{buhrman2002bitvectors} which is based on a nonexplicit construction of expander graphs. 
	
	All the above papers on the bit-probe model consider fixed-length block coding. Variable-length source coding was investigated by Pananjady and Courtade~\cite{pananjady2018effect} who gave upper and lower bounds on the achievable rate for the compression of sparse sequences under local decodability constraints. 
	
	The works~\cite{vatedka_local_2020,vatedka2020log} considered simultaneous local decodability and update efficiency. In particular,~\cite{vatedka_local_2020} designed a compressor whose average-case local decodability (defined as the expected number of bits that need to be probed to recover any $X_i$) and the average-case update efficiency (the expected number of bits that need to be read and written in order to update a single $X_i$) both scale as $O\left( \frac{1}{\varepsilon^2}\log \frac{1}{\varepsilon} \right)$. In fact, our scheme for distributed compression with locality is inspired by the multilevel compression scheme in~\cite{vatedka_local_2020}. The paper~\cite{vatedka2020log} designed a compression scheme whose worst-case local decodability and update efficiency scales as $O(\log\log n)$. More recently,~\cite{vestergaard2021enabling,kamparaju2022low} implemented different versions of the concatenation scheme and evaluated its performance on practical datasets.

\subsection{Paper organization}
In Section~\ref{pre}, we introduce notions of localities and formally define the problem. In Section~\ref{mainresultato}, we present our results. In Sections~\ref{prova1},~\ref{confusable_below}, and \ref{sec:achievability} we prove the results, and in Section~\ref{estensione}, we discuss the extension to more than two sources. In Section~\ref{remfin}, we draw concluding remarks. 

\section{Preliminaries and Problem Statement}\label{pre}

\subsection{Distributed compression without locality}\label{sces}
Let $(X^n,Y^n)$ be $n$ independent copies of a pair of random variables $(X,Y)\sim p_{XY}$ defined over some finite alphabet ${\cal{X}}\times {\cal{Y}}$, with $|{\cal{X}}|\geq 2, |{\cal{Y}}|\geq 2$. Without loss of generality, we assume that  ${\cal{X}}=\{x: p_X(x)>0\}$ and ${\cal{Y}}=\{y: p_Y(y)>0\}$. 

Sequences $X^n$ and $Y^n$ represent two sources of information separately encoded into binary codewords $C^{nR_1}$ and $C^{nR_2}$ at rates $R_1$ and $R_2$, respectively. Upon receiving these codewords, a receiver outputs the sources estimates $(\hat{X}^n,\hat{Y}^n)$ and makes an error with probability
$$P_e \defeq \Pr[(\hat{X}^n,\hat{Y}^n)\ne ({X}^n,{Y}^n)].$$

The rate region is the closure of the set of rate pairs $(R_1,R_2)$ for which $P_e\to 0$ as $n\to \infty$, and is given by:
\begin{theorem}[Slepian-Wolf, \cite{slepian1973noiseless}, \cite{csiszar1980towards}]\label{thm:slepianwolf}
The rate region of a source $p_{XY}$ is the set of pairs $(R_1,R_2)$ that satisfy
\begin{align}
  R_1 &\geq H(X|Y) &\notag\\
  R_2 &\geq H(Y|X) &\notag\\
  R_1+R_2 &\geq H(X,Y). &\label{eq:sw_region}
\end{align}
Moreover, for any $(R_1,R_2)$ in the interior of the rate region, and any $\varepsilon >0$, there exist a sequence of coding schemes operating at rates at most $ R_1+\varepsilon$ and $R_2+\varepsilon$ such that
\[
  P_e \defeq \Pr[(\hat{X}^n,\hat{Y}^n)\neq (X^n,Y^n)] \leq 2^{-n(E-\varepsilon )},
\]
where $E$ is a constant that is specified by $R_1,R_2$ and~$p_{XY}$. 
\end{theorem}

\subsection{Distributed compression with locality}\label{sec:statements_results}
\subsubsection{Local decoder}
Given encodings $C^{nR_1}$ and $C^{nR_2}$, a local decoder takes as input $i\in [n]$,\footnote{$[n] \defeq \{1,2,\ldots, n\}$}  probes/reads a fixed set $\cI_i$ of components from $C^{nR_1}$ and $C^{nR_2}$, which we  denote as $C_{\cI_i}$, and outputs an estimate $(\hat{X}_i,\hat{Y}_i)$ of $(X_i,Y_i)$.
The worst-case local decodability and error probability are defined as
\begin{equation}\label{eq:def_wc_localdecodability}
\rwc \defeq \max_{1\leq i\leq n}d(i),
\end{equation}
where $d(i)\defeq |\cI_i |$, and 
$$\peld \defeq \max_{1\leq i\leq n}\Pr[(\hat{X}_i,\hat{Y}_i)\neq(X_i,Y_i)].$$
Note that $\cI_i$ may contain  different sets of components from $C^{nR_1}$ and $C^{nR_2}$, but these components should be chosen non-adaptively; conditioned on the index $i$, set $\cI_i$ should be independent of $(X^n,Y^n)$. Note also that a sequence of $\rwc$ adaptive (random) queries takes at most $2^{\rwc}$ different values.\footnote{A sequence of $\rwc$ random adaptive queries can be represented as a complete binary decision tree of depth $\rwc$, where any node (including the root and the leaves) is labelled with a codeword component (among the $n(R_1+R_2)$ possible), and where each edge is labelled $0$ or $1$. Any instance of $\rwc$ adaptive queries describes one of the $2^\rwc$ path from the root to a leaf.} Therefore, a lower bound on the probability of error for locality-$\rwc$ nonadaptive decoders (the main contribution of this paper), translates into a lower bound for locality-$\log \rwc$ adaptive decoders. Finally, notice that even though $\cI_i$ is non-adaptively chosen, it could still be a random set, in which case $d(i)$ is defined as the essential supremum of $|\cI_i|$.

{\remark Note that the notation ${\cI_i}$ leaves out any reference to the underlying sources. In particular, if both sources $X^n$ and $Y^n$ are compressed, then the set ${\cI_i}$ may contain coordinates from both $C^{nR_1}$ and $C^{nR_2}$, and if only source $X^n$ is compressed, then the set ${\cI_i}$ contains components from $C^{nR_1}$ only. }
\subsubsection{Strong vs. weak locality}\label{sweak}
 A rate pair $(R_1,R_2)$ is said to be achievable with \emph{strong locality} if $$\peld = o(1)\quad \text{and}\quad \rwc = \Theta(1)\quad \text{as}\:\: n \to \infty.$$ That is, by probing only a constant number (independent of~$n$) of symbols, the error probability of the local decoder goes to zero as the blocklength increases.  By contrast, $(R_1,R_2)$ is said to be achievable with \emph{weak locality} if $$\peld=o(1)\quad \text{and} \quad \rwc=\omega(1)\quad \text{as}\:\: n\to \infty.$$ 
 
 Weak locality is always achievable through the concatenation scheme where source sequences $X^n$ and $Y^n$ are decomposed into length $b$ sequences $$X^b(j)\defeq X_{(j-1)b+1}^{jb}\defeq (X_{(j-1)b+1},X_{(j-1)b+1},\ldots,X_{jb})$$  $$Y^b(j)\defeq Y_{(j-1)b+1}^{jb}\defeq (Y_{(j-1)b+1},Y_{(j-1)b+1},\ldots,Y_{jb})$$ for $j=1,2,\ldots$ and each block $(X^b(j),Y^b(j))$ is independently compressed using a Slepian-Wolf code operating at the desired $(R_1,R_2)$. Given $i\in[n]$, the local decoder decodes block $j=\lceil i/b\rceil$ (thereby reading $b(R_1+R_2)$ compressed bits), and outputs the estimates of the $i$-th bit of $X^n$ and $Y^n$. By letting $\rwc=b(R_1+R_2)$ in Theorem~\ref{thm:slepianwolf} we get:
\begin{corollary}[Concatenation]
\label{lemma:concatenation}
For any source $p_{XY}$ and any $(R_1,R_2)$ in the interior of the rate region \eqref{eq:sw_region}, the concatenation scheme achieves weak locality:
  \[
    \peld \leq 2^{-\Theta(\rwc)}.
  \]
\end{corollary}
 
 \subsection{Statement of the problem}
 By contrast with weak locality, whether strong locality is generally achievable is much less clear. In fact, it is only recently that strong locality was shown to be achievable for the single source setup at any lossless compression rate $R>H(X)$  \cite{mazumdar2015local,tatwawadi18isit_universalRA}. For the Slepian-Wolf setup at hand, this result implies  that strong locality holds for any  $(R_1,R_2)$ such that $R_1>H(X)$ and $R_2> H(Y)$. In this regime, sources can be encoded using the single source strongly local codes of \cite{mazumdar2015local,tatwawadi18isit_universalRA}, separately for source $X^n$ and source $Y^n$---and ignore dependency between $X^n$ and $Y^n$. Does this conclusion extend to the regime where at least one of the sources is encoded at a rate below its entropy?

\section{Main results}\label{mainresultato}
Our main result answers the above question in the negative: if the source is ``confusable'', strong locality is impossible whenever one of the sources is compressed below its entropy.

\begin{definition}[Source confusability]
Source  $p_{XY}$ is said to be $\cX$-confusable if for every $x_1,x_2\in {\cal{X}}$, there exists $y\in {\cal{Y}}$ such that $p_{X|Y}(x_1|y)>0$ and $p_{X|Y}(x_2|y)>0$---recall that $p_Y(y)>0$ for any $y\in~{\cal{Y}}$, see Section~\ref{sces}.
\end{definition}
Any source with full support, {\it{i.e.}}, such that $p_{XY}(x,y)>0$ for all $(x,y)\in {\cal{X}}\times {\cal{Y}}$, is both $\cX$- and $\cY$-confusable. An example of an $\cX$-confusable source which does not have full support is $p_{XY}$ where $p_Y=$ Bernoulli($p$), $0<p<1$, and where $p_{X|Y}$ is a $Z$ channel with crossover parameter $0<\varepsilon<1$. Instead, if $p_{X|Y}$ is the erasure channel, source $p_{XY}$ is not $\cX$-confusable. If ${\cal{X}=\cal{Y}}=\{0,1\}$, source $p_{XY}$ is always $\cX$-confusable except if $p_{X|Y}$ is the noiseless channel.

\begin{theorem}\label{thm:converse_general}
  Suppose source $ p_{XY}$ is $\cX$-confusable. Suppose $X^n$ is encoded into codeword $C^{nR_1}$ with $R_1<H(X)$, and suppose the code has a local decoder with worst-case local decodability $\rwc\in [n]$. Then 
  $$\max_{1\leq i\leq n}\Pr[\hat{X}_i(C_{{\cI}_i},Y^n)\neq X_i]\geq 2^{-\Theta(\rwc)},$$
  where $\hat{X}_i(C_{{\cI}_i},Y^n)$ denotes any estimator of source symbol $X_i$ given observations $C_{{\cI}_i}$ and $Y^n$. 
\end{theorem}
 This result says that if the source is $\cX$-confusable, then strong locality is impossible whenever $R_1<H(X)$; not even the full cooperation of the $Y$-transmitter through the uncompressed source $Y^n$ allows to achieve strong locality. A particular version of this theorem for doubly symmetric sources, where $p_X$ is the Bernoulli($1/2$) distribution and where $p_{Y|X}$ corresponds to a BSC($\rho$) for some crossover parameter $0<\rho<1/2$, was proved in~\cite{vatedka2022_conferenceversion}. Finally note that for adaptive probing the lower bound given in Theorem~\ref{thm:converse_general} becomes
 $$\max_{1\leq i\leq n}\Pr[\hat{X}_i(C_{{\cI}_i},Y^n)\neq X_i]\geq 2^{-\Theta(2^\rwc)}.$$
Hence, if the source is $\cX$-confusable and if $R_1<H(X)$ strong locality cannot be achieved even under adaptive probing.
The $\cX$-confusability property turns out to be necessary for Theorem~\ref{thm:converse_general} to hold:
\begin{theorem}\label{lemma:yconfusable_necessary}
  Suppose source $p_{XY}$ is not $\cX$-confusable.
  Then, strong locality is achievable at some $R_1<H(X)$ and $R_2=H(Y)$.
\end{theorem}

From Corollary~\ref{lemma:concatenation}, for any $(R_1,R_2)$ in the interior of the rate region the concatenation scheme achieves a local error probability that decays as $2^{-\Theta(\rwc)}$, and this is order optimal by Theorem~\ref{thm:converse_general} for confusable sources. However, note that the local decoding error-probability of the concatenation scheme is tied to a specific value of $\rwc$ which is equal to the sub-block length $b$. In particular, if $b=\Theta(1)$, then, because the concatenation scheme encodes each sub-block independently, it is impossible to recover $(X^n,Y^n)$ with vanishing probability of error as $n$ grows, even after probing the entire compressed sequences(!) To lower the error-probability, the parameter $b$, hence the encoding procedure, should be modified accordingly.

Our second contribution is a compression scheme whose local decoder has an error probability that decreases as the number of probed symbols increases, without changing the encoding. The performance of this scheme is given in the following theorem:
\begin{theorem}\label{thm:achievability}
  For any $(R_1,R_2)$ in the interior of the rate region, there exists a rate $(R_1,R_2)$  encoder and a local decoder such that for every $1>\eta>2^{-2^{O(\log n)}}$ the local decoder achieves $\peld\leq \eta$ while probing $\rwc=\mathrm{poly}(\log(1/\eta))$ bits.
\end{theorem}
Theorems~\ref{thm:converse_general},\ref{lemma:yconfusable_necessary}, and~\ref{thm:achievability} easily generalize to more than two sources, see Section~\ref{estensione}.

\emph{Note:} The present paper differs from the ISIT paper \cite{vatedka2022_conferenceversion} mainly in that it establishes the impossibility of strong locality (Theorem~\ref{thm:converse_general}) for the most general class of sources (confusable sources), and not only for the specific class of doubly symmetric binary sources. In fact, the arguments used in \cite{vatedka2022_conferenceversion} do not extend beyond sources with full-support. The arguments used here are not only more general, but also more direct than those in \cite{vatedka2022_conferenceversion}. Theorem~\ref{lemma:yconfusable_necessary} is new and \cite{vatedka2022_conferenceversion} contains mostly a sketch of the proof of Theorem~\ref{thm:achievability}. Theorems~\ref{ext1}, \ref{ext2}, and \ref{ext3} that extend the above results to more than two sources (see Section~\ref{estensione}) did not formally appear in \cite{vatedka2022_conferenceversion}.

\section{Proof of Theorem~\ref{thm:converse_general}}\label{prova1}
\subsection{Preliminaries}

One key element in proving Theorem~\ref{thm:converse_general} is the following coupling. Given $p_{X,Y}$ define random variable $\widetilde{X}$ so that $$X-Y-\widetilde{X}$$ forms a Markov chain and so that $$p_{X|Y}=p_{\widetilde{X}|Y}.$$ 
Observe that if $p_{XY}$ is $\cX$-confusable, then for any given $(x,\widetilde{x})\in {\cal{X}}\times {\cal{X}}$ there exists $y$, with $p_Y(y)>0$ (recall that without loss of generality $p_Y(y)>0$ for any $y\in {\cal{Y}}$), such that $$p_{X,\widetilde{X}}(x,\widetilde{x}|y)=p_{X|Y}(x|y)p_{\widetilde{X}|Y}(\wtx |y)>0.$$ Hence, we have:
\begin{lemma}\label{lemma:pxx'}
 If $p_{XY}$ is $\cX$-confusable, then $p_{X\wtX}$ has full support.
\end{lemma}
In turn, since distributions with full support (and finite alphabet) are reverse hypercontractive \cite[Theorem 1]{kamath_reverse_2015}, we get:
\begin{lemma}\label{lemma:reversehypercontractivity}
 If $p_{XY}$ is $\cX$-confusable then, 
  for every $\cA,\cB\subset {\cal{X}}^n$, we have
  \[
    \Pr[X^n\in\cA,\widetilde{X}^n\in\cB] \geq \left( \Pr[X^n\in\cA] \right)^{\alpha} \left( \Pr[\widetilde{X}^n\in\cB] \right)^{\beta}
  \]
  for some finite constants $\alpha,\beta$.\footnote{More precisely, $\alpha$ and $\beta$ are in $(1,\infty)$ (see \cite{kamath_reverse_2015}), but for our purpose the values of $\alpha$ and $\beta$ (as functions of $p_{X\widetilde{X}}$) are irrelevant.}
\end{lemma}
The other key element in proving Theorem~\ref{thm:converse_general} is the following general lemma:
\begin{lemma}\label{lemma:pmap_pxx'}
Fix source $p_{XY}$. Suppose $X^n$ is encoded into codeword $C^{nR_1}$ at some rate $R_1\geq 0$. Fix $i\in [n]$ and let $\hat{X}_i(C_{\cI_i},Y^n)$ be an estimator of $X_i$ given $C_{\cI_i}$ and $Y^n$. Then, for any realization $c$ of $C_{\cI_i}$, we have:
 \begin{align}
  \Pr&(\hat{X}_i(C_{\cI_i},Y^n)\ne X_i,C_{\cI_i}=c) \notag\\
  &\geq \Pr[X_i = \bar{x},{C}_{\cI_i}=c,\widetilde{X}_i\ne \bar{x},\tilde{C}_{\cI_i}=c]
  \notag \end{align}
where  
\begin{align}\label{xmax}
\bx \defeq \arg\max_{x}\Pr[X_i=x|C_{\cI_i}=c],
\end{align}
and where $\tilde{C}_{\cI_i}$ is obtained by encoding $\widetilde{X}^n$ with the same code as for $X^n$. 
\end{lemma}
The last ingredient for proving Theorem~\ref{thm:converse_general} is the following result which follows from a basic rate-distortion argument:
\begin{lemma}
 Suppose $X^n$ is encoded into codeword $C^{nR_1}$ at some rate $R_1<H(X)$. Suppose the code has a local decoder with worst-case local decodability $\rwc=\max_{1\leq i\leq n} d(i) \in [n]$.	Then, there exists a constant $\delta>0$ that depends only on $R$ (and $p_X$), an index $i\in [n]$, and a realization $c$ of $C_{\cI_i}$ such that
\[
\Pr[\hat{X}_{i}(C_{\cI_i})\neq X_i,C_{\cI_i}=c] \geq \delta 2^{-{\rwc}}.
\]
	\label{lemma:worsti}
\end{lemma}
Lemmas~\ref{lemma:pmap_pxx'} and ~\ref{lemma:worsti}
are proved in Section~\ref{apendiccio}.
\subsection{Proof of Theorem~\ref{thm:converse_general}}\label{sec:general_converse_nonhypercontractive}
  Suppose the source is confusable and suppose $X^n$ is compressed at rate $R_1<H(X)$.
Using Lemma~\ref{lemma:pmap_pxx'} then Lemma~\ref{lemma:reversehypercontractivity}, we have that for some finite constants $\alpha$ and $\beta$, any index $i\in [n]$, and any realization $c$ of $C_{\cI_i}$
\begin{align}\label{alianza}
  &\Pr(\hat{X}_i(C_{\cI_i},Y^n)\ne X_i,C_{\cI_i}=c) \notag\\
  &\geq (\Pr[X_i = \bar{x},C_{\cI_i}=c])^{\alpha}\notag\\
    &\hspace{3cm}\times (\Pr[\widetilde{X}_i\ne \bar{x},\tilde{C}_{\cI_i}=c])^{\beta}
 \end{align}
 where $\bar{x}$ is defined in \eqref{xmax}.
 
As a last step, we now show that, for some index $i\in [n]$ and some realization $c\in \{0,1\}^{d(i)}$, each of the two probability terms on the right-hand side of the inequality \eqref{alianza} is lower bounded by $2^{-\Theta(\rwc)}$. By summing both sides of inequality \eqref{alianza} over $c$'s, we then deduce that 
 \[
 \Pr(\hat{X}_{i}(C_{\cI_{i}},Y^n)\ne X_{i}) \geq 2^{-\Theta(\rwc)}
 \]
for some $i\in [n]$, thereby completing the proof of Theorem~\ref{thm:converse_general}.

Let us start with the second term. Since $\widetilde{X}^n$ has the same distribution as  ${X}^n$, from Lemma~\ref{lemma:worsti} 
there exist a constant $\delta>0$, an index $i\in  [n]$, and a local codeword $c\in \{0,1\}^{d(i)}$, with $d(i)\leq \rwc$, such that
\begin{equation}\label{eq:yi_notequals_a_2}
\Pr[\widetilde{X}_{i}\neq \bx,\tilde{C}_{\cI_{i}}=c] \geq \delta 2^{-{\rwc}}\geq \delta |\cX|^{-\rwc}, 
\end{equation}
where the second inequality in \eqref{eq:yi_equals_a_2} holds since $|{\cal{X}}|\geq 2$ (see Section~\ref{sces}).

For the first term, note that
\begin{equation}\label{eq:yi_equals_a_2}
\Pr[X_i= \bx|C_{\cI_{i}}=c] \geq \frac{1}{|\cX|}, 
\end{equation}
for otherwise the probabilities would not sum to one. From \eqref{eq:yi_notequals_a_2} and \eqref{eq:yi_equals_a_2} it then follows that
\begin{align}\label{ineqal_2}
\Pr[X_{i}&= \bx,C_{\cI_{i}}=c]\notag\\
&=\Pr[X_{i}= \bx|C_{\cI_{i}}=c]\Pr[C_{\cI_{i}}=c]\notag\\
&\geq  \frac{1}{|\cX|}\delta |\cX|^{-\rwc}\notag\\
&\geq \delta |\cX|^{-\rwc-1}. \end{align}
This establishes the desired claim.\hfill \qed
\subsection{Proofs of Lemmas~\ref{lemma:pmap_pxx'} and \ref{lemma:worsti}}\label{apendiccio}

\begin{proof}[Proof of Lemma~\ref{lemma:pmap_pxx'}]
 For any estimator $\hat{X}_i(C_{\cI_i},Y^n)$ of $X_i$, we have 
  \begin{align}\label{nabucco}
   &\Pr(\hat{X}_i(C_{\cI_i},Y^n)\ne X_i,\; C_{\cI_i}=c) \notag\\
   &\geq \Pr(\hat{E}(C_{\cI_i},Y^n)\ne E(X_i),\; C_{\cI_i}=c)\nonumber \\
    &=\sum_{y^n}\Pr(\hat{E}(C_{\cI_i},Y^n)\ne E(X_i), \; C_{\cI_i}=c |Y^n = y^n)\nonumber\\
    &\hspace{3cm}\times  \Pr(Y^n=y^n)\nonumber \\
     &\geq  \sum_{y^n} \Pr(Y^n=y^n)\nonumber\\ &\times \min \Big\{ \Pr[X_i= \bar{x},\; C_{\cI_i}=c|Y^n = y^n], \notag\\
     &\hspace{1.2cm}\Pr[X_i\ne\bar{x},\; C_{\cI_i}=c|Y^n = y^n] \Big\}\,,
  \end{align}
  where $\hat{E}(C_{\cI_i},Y^n)$ is an estimator of the binary random variable $E(X_i)$, defined to be equal to zero if $X_i=\bx$ and one if $X_i\ne \bx$; and where the right-hand side of the second inequality is the error probability of the optimal (MAP) estimator $\hat{X}$ with the foreknowledge of $Y^n$.

 By multiplying the minimum on the right-hand side by the maximum of the same terms (which is at most one), we get
  \begin{align*}
     &\Pr(\hat{X}_i(C_{\cI_i},Y^n)\ne X_i,\; C_{\cI_i}=c) \\
     &\geq \sum_{y^n} \Pr(Y^n=y^n) \bigg(\min\{ \cdot \} \times\max\{ \cdot \} \bigg)\\
     &= \sum_{y^n} \Pr(Y^n=y^n) \Pr[X_i=\bar{x},C_{\cI_i}=c|Y^n = y^n]\\
     &\hspace{2.7cm}\times \Pr[X_i\ne\bar{x},C_{\cI_i}=c|Y^n = y^n] \\
 &= \sum_{y^n} \Pr(Y^n=y^n) \Pr[X_i= \bar{x},C_{\cI_i}=c|Y^n = y^n]\\
    &\hspace{2.7cm} \times \Pr[\widetilde{X}_i= \bar{x},\tilde{C}_{\cI_i}=c|Y^n = y^n]\\
    &= \Pr[X_i = \bar{x},C_{\cI_i}=c, \widetilde{X}_i\ne \bar{x},\tilde{C}_{\cI_i}=c]
    \end{align*}
where the second equality holds since $p_{X|Y}=p_{\widetilde{X}|Y}$. This yields the desired result.
\end{proof}
\begin{proof}[Proof of Lemma~\ref{lemma:worsti}]
The converse to Shannon's lossy source coding theorem implies that if $R<H(X)$, then there exists a $\delta=\delta(R)>0$ such that 
	\[
	\bE d_H(X^n,\hat{X}^n) = \sum_{i=1}^{n}\Pr[\hat{X}_{i}(C_{\cI_i})\neq X_i ]\geq n \delta
	\]
	where $d_H(X^n,\hat{X}^n)$ denotes the Hamming distance between $X^n$ and $\hat{X}^n$.
	Hence,  $$\delta \leq \Pr[\hat{X}_{i}(C_{\cI_i}) \neq X_i ]$$
	for at least one index $i\in [n]$.
	
Expanding the right-hand side and assuming a worst-case local decodability  of $\rwc \in [n]$, we have
\begin{align*}
\delta &\leq \Pr[\hat{X}_{i}(C_{\cI_i})\neq X_i]\\
       &= \sum_{c\in \{0,1\}^{d(i)}} \Pr[\hat{X}_{i}(C_{\cI_i})\neq X_i,C_{\cI_i}=c]\\
       &\leq 2^{\rwc}\max_{c\in\{0,1\}^{d(i)}} \Pr[\hat{X}_{i}(C_{\cI_i})\neq X_i,C_{\cI_i}=c]
\end{align*}
which concludes the proof.
\end{proof}

\section{Proof of Theorem~\ref{lemma:yconfusable_necessary}}\label{confusable_below}
  If $p_{XY}$ is not $\cX$-confusable, then there exists $x_1,x_2\in \cX$ such that, for any $y\in \cY$,  either $p_{XY}(x_1,y)>0$ or $p_{XY}(x_2,y)>0$ (recall that without loss of generality $p_Y(y)>0$ for any $y\in {\cal{Y}}$). Therefore, conditioned on $X\in \{x_1,x_2\}$, the knowledge of $Y$ reveals $X$.
  
 Let $\cX=\{1,2,\ldots,|{\cal{X}}|\}$, and suppose without loss of generality that $(x_1,x_2)=(1,2)$. Define the new source $U^n$ over the reduced alphabet $\{2,3,\ldots,|{\cal{X}}| \}$ as
\[
  U_i = \begin{cases}
    2 &\text{ if }X_i\in \{1,2\}\\
     X_i &\text{ if }X_i\neq \{1,2\}.\\
  \end{cases}
\]
Clearly, $H(U)<H(X)$ and $(U,Y)$ determines $(X,Y)$. We can therefore compress $U^n$ and $Y^n$ independently at rates $R_1=H(U)<H(X)$ and $R_2=H(Y)$ using the compressors of~\cite{tatwawadi18isit_universalRA,mazumdar2015local} to achieve strong locality.

\section{Proof of Theorem~\ref{thm:achievability}}\label{sec:achievability}
We want a scheme that achieves the following: For any fixed $\delta>0$ and $(R_1,R_2)$ within the Slepian-Wolf rate region,
\begin{itemize}
\item The sequences $(X^n,Y^n)$ are independently compressed to rates $(R_1+\delta,R_2+\delta)$ respectively.
\item For any $i\in[n]$ and $1>\eta>2^{-2^{O(\log n)}}$ specified at the receiver, the local decoder probes $\mathrm{poly}(\log(1/\eta))$ compressed bits, and outputs $(\hat{X}_i,\hat{Y}_i)$ which satisfies
\[
\peld = \Pr[(\hat{X}_i,\hat{Y}_i)\neq (X_i,Y_i)] \leq \eta.
\]
\end{itemize}

Our coding scheme is inspired by that in~\cite{vatedka_local_2020}, and is a hierarchical compression scheme. The compressed bits consist of various blocks that are spread across multiple ``levels'' $1\leq \ell\leq \ell_{\max}$. The compressed bits at level $\ell=0$ is obtained by applying the concatenation scheme defined in Section~\ref{sweak} with $b=O(1)$. This guarantees that any pair of source symbols can be recovered with $2^{-\Theta(b)}=O(1)$ probability of error. The compressed bits at higher levels $\ell\geq 1$ can be viewed as additional refinement bits that are probed only when we desire a lower probability of error. By probing blocks corresponding to higher levels, we obtain a more reliable estimate of $(X_i,Y_i)$. The compressed blocks at level $\ell$ are obtained by using a random binning scheme applied to blocks of size $n_\ell$, where $n_{\ell}$ is growing superexponentially with $\ell$. However, the rates for higher levels is chosen to decay exponentially with $\ell$.  The key challenge is to choose the parameters carefully so that the additional bits corresponding to higher levels provide a negligible contribution to the overall compression rates.

\subsubsection{Parameters}
We choose\footnote{Since we only aim to get order-optimal results, we have not attempted to optimize over the various parameters.} a sufficiently small $\varepsilon_0>0$, positive integers
\begin{equation}\label{eq:ach_prf_bzero}
 b_0 = n_0 
\end{equation}
which are constants independent of $n$, and
\begin{equation}\label{eq:ach_prf_kzero}
\kza = \lceil (R_1+\varepsilon_0)b_0\rceil \text{ and }\kzb = \lceil (R_2+\varepsilon_1)b_1\rceil 
\end{equation}
such that the probability of error of a Slepian-Wolf code for sequences of length $b_0$ satisfies
$$\Pr[(\hat{X}^{b_0},\hat{Y}^{b_0})\neq (X^{b_0},Y^{b_0})]\leq 2^{-\beta \varepsilon_0 b_0}\leq \delta'$$ where $\beta>0$  depends on $p_{XY},R_1,R_2$ only and $0<\delta'<1$ is a parameter that determines an upper bound on the probability of local decoding error that can be achieved. 

For each $\ell=1,2,\ldots,\lmax$, define
\begin{align}
  \vel &= \velm/2 = \epsilon_0/2^\ell &\notag\\
  b_\ell &= 16b_{\ell-1} = 16^\ell b_0 &\notag\\
  n_\ell &= b_\ell n_{\ell-1} = 4^{\ell(\ell+1)}b_0^{\ell+1} & \notag\\
  \kla &= \vel n_{\ell}\left(\beta+|\cX| +\frac{b_\ell}{n_\ell}\log \frac{e2^\ell}{\varepsilon_0} \right) & \notag\\
 \klb &= \vel n_{\ell}\left(\beta+|\cY| +\frac{b_\ell}{n_\ell}\log \frac{e2^\ell}{\varepsilon_0} \right) & \label{eq:ach_prf_params}
\end{align}

\subsubsection{Codes for various levels}
At the heart of our construction is a multilevel random binning argument that can be described by a sequence of random codes $\cC_0,\cC_1,\ldots,\cC_{\lmax}$. 

The code $\cC_\ell$ at level $\ell$ consists of two encoders $\Phi_\ell:\cX^{n_\ell}\to \{0,1\}^{\kla}$ and $\Psi_\ell:\cY^{n_\ell}\to\{0,1 \}^{\klb}$, where $\kla$ and $\klb$ are as defined previously. For each $x^{n_\ell}\in \cX^{n_{\ell}}$, we assign a codeword $\Phi_{\ell}(x^{n_\ell})$ drawn uniformly at random from $\{0,1\}^{\kla}$. Similarly, for each $y^{n_\ell}\in \cY^{n_{\ell}}$, we assign a codeword $\Psi_{\ell}(y^{n_\ell})$ drawn uniformly at random from $\{0,1\}^{\klb}$. For any $u^{\kla}\in \{0,1\}^{\kla}$ and $v^{\klb}\in\{0,1\}^{\klb}$, we have $\Pr[\Phi_\ell(x^{n_\ell})=u^{\kla}]=2^{-\kla}$ and $\Pr[\Phi_\ell(x^{n_\ell})=v^{\klb}]=2^{-\klb}$. The codes are known to the decoder and the respective encoders.


\subsubsection{Encoder}
We now describe the encoding of the sequences $X^n$ and $Y^n$. Let us suppose that user 1 has $X^n$ and user 2 has $Y^n$. 
 
The codeword generated by each user comprises of various blocks spread over multiple levels, and the encoding is done independently at each level.
Consider any level $0\leq \ell\leq \lmax$. Each user partitions its source sequence into blocks of $n_\ell$ symbols each. For $i=1,2,\ldots,n/n_\ell$, define the $i$th (source) block at level $\ell$ to be $X^{n_\ell}(\ell,i)=X_{(i-1)n_\ell+1}^{in_\ell}$, and $Y^{n_\ell}(\ell,i)=Y_{(i-1)n_\ell+1}^{in_\ell}$. Let $U^{\kza}(\ell,i)\triangleq \Phi_\ell(X^{n_\ell}(\ell,i))$ be the $i$'th level-$\ell$ codeword for user 1, and $V^{\kza}(\ell,i)\triangleq \Psi_\ell(Y^{n_\ell}(\ell,i))$ be the $i$'th level-$\ell$ codeword for user 2.

The codeword for $X^n$ is obtained by taking the concatenation of all level $\ell$ codewords for $0\leq \ell \leq \lmax$. This is equal to $(U^{\kza}(\ell,i):0\leq \ell\leq i, 1\leq i\leq n/n_\ell)$. Similarly, the codeword for $Y^n$ is equal to $(V^{\kzb}(\ell,i):0\leq \ell\leq i, 1\leq i\leq n/n_\ell)$.

An illustration of the encoding process is provided in Fig.~\ref{fig:encoding_user1}.
The level-$0$ codewords correspond to the concatenation scheme, and most of the entropy of the compressed sequence lies in the level-0 codewords.  The level $\ell\geq 1$ codewords give extra information that allow us to reduce the probability of local decoding error. The rates $\kla/n_{\ell}$ and $\klb/n_{\ell}$ are exponentially decaying functions of $\ell$, and the overall sum rates of all the level $\ell=1,2,\ldots,\lmax$ codewords is negligible.

\begin{figure*}
    \centering
    \includegraphics[width=0.9\textwidth]{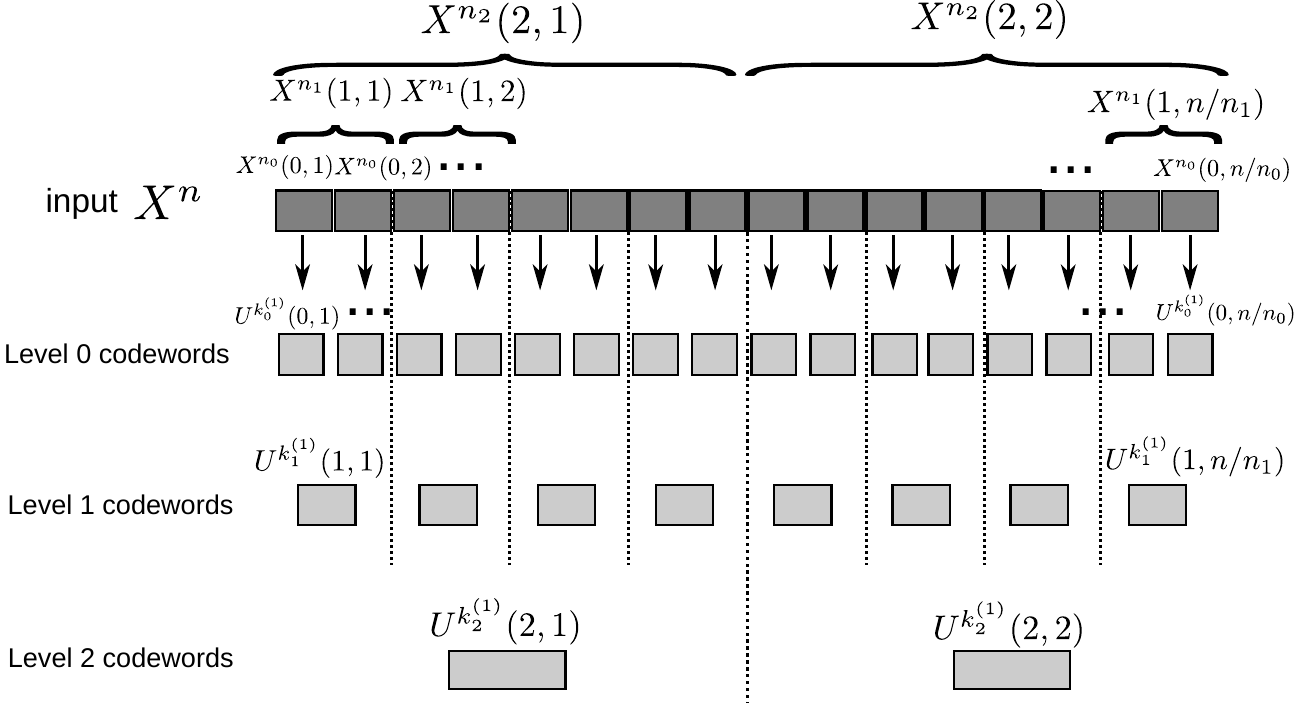}
    \caption{A depiction of the encoding structure for user 1. For ease of illustration, we have chosen $\lmax=2$, $b_2=4b_1$, and $b_1=2b_0$. To obtain the codewords at level $\ell$, the source $X^n$ is partitioned into blocks of $n_\ell$ symbols each, and the $j$'th level-$\ell$ codeword $U^{\kla}(\ell,j)=\Phi_\ell(X^{n_\ell}(\ell,j))$. The overall codeword is the concatenation of all the level $\ell=0,1,\ldots,\lmax$ codewords. A similar encoding process is performed by user $2$ but using $\Psi_\ell$. Consider decoding $X^{n_0}(0,2)$. The local decoder may choose to probe $U^{\kza}(0,2),V^{\kzb}(0,2)$ and use the standard Slepian-Wolf (joint typicality) decoder. If a lower probability of error is required, then it additionally probes $U^{\kla}(\ell,0),V^{\klb}(\ell,b)$ for $\ell=1$, or $\ell=1,2$ depending on the target probability of error. The additional bits are then used to refine the estimate of $X^{n_0}(0,2)$.}
    \label{fig:encoding_user1}
\end{figure*}

\subsubsection{Local decoder}
The local decoder takes two parameters as input: a location $i\in [n]$, and $\ell_d\in \{0,1,2,\ldots,\lmax\}$. The first parameter specifies which $(X_i,Y_i)$ the decoder wishes to recover. The second parameter specifies the number of bits to probe (which decides the probability of error). For a specified $\ell_d$, the local decoder probes $2^{O(\ell_d^2)}$ compressed bits, and the probability of error is $2^{-2^{\tilde{O}(\ell_d^2)}}$. This statement will be made more precise shortly.

The decoder works by probing compressed bits up to level $\ell_d$ as follows:
\begin{itemize}
\item The decoder first finds which level-$\ell_d$ chunk the desired location $i$ lies in. In other words, it sets $i_{\ell_d}=\lceil i/n_{\ell_d}\rceil$. It then reads all the compressed chunks up to level $\ell_d$ corresponding to the symbols $X_{(i_d-1)n_{\ell_d}+1}^{i_dn_{\ell_d}}$.
\item The decoder now iteratively improves its estimate by processing the compressed bits from level $0$ to level $\ell_d$ as follows:
  \begin{itemize}
  \item At level 0, the decoder uses the Slepian-Wolf decoder to obtain the level-0 estimates of $X_{(i_d-1)n_{\ell_d}+1}^{i_dn_{\ell_d}}$. Call this estimate as $\hat{X}_{(i_d-1)n_{\ell_d}+1}^{i_dn_{\ell_d}}(0)$.
  \item For all subsequent levels $\ell\in\{1,2,\ldots,\ell_d\}$, the decoder does the following. Suppose that for some $i$, we want to estimate $(X^{n_{\ell}}(\ell,i),Y^{n_{\ell}}(\ell,i))$ assuming that we already have the level $\ell-1$ estimate\footnote{The level $\ell-1$ estimate is obtained by decoding all the compressed bits up to level $\ell-1$ corresponding to $(X^{in_{\ell}}_{(i-1)n_{\ell}+1},Y^{in_{\ell}}_{(i-1)n_{\ell}+1})$.}. The decoder outputs $(\hat{X}^{{n_{\ell}}}(\ell,i),\hat{Y}^{{n_{\ell}}}(\ell,i))=(x^{n_{\ell}},y^{n_{\ell}})$ if $(x^{n_{\ell}},y^{n_{\ell}})$ is the unique pair of sequences which match $U^{\kla}(\ell,i),V^{\klb}(\ell,i)$ and also match at least $(1-\vel) b_{\ell}$ of the level-$(\ell-1)$ estimated blocks. If there is no such unique sequence, then the level-$\ell$ decoder outputs the zero sequence.
  \end{itemize}

\end{itemize}

In Lemma~\ref{lemma:ach_localdecodability}, we derive an upper bound on the local decodability. We then derive an upper bound on the probability of local decoding error in Lemma~\ref{lemma_ach_pe}. Combining the two gives us Theorem~\ref{thm:achievability}.

\begin{lemma}\label{lemma:ach_localdecodability}
  For any given parameters $(i,\ell_d)$, the number of bits probed by the local decoder is
  \[
    \rwc(\ell_d) \leq  b_0^{\ell_d+1}4^{\ell_d(\ell_d+1)}(R_1+R_2+\gamma_1\varepsilon_0) \leq 2^{\gamma_2\ell_d^2}
  \]
\end{lemma}
where $\gamma_1$ is a constant that only depends on $p_{XY},R_1,R_2$, while $\gamma_2$ may depend on $p_{XY},R_1,R_2,\varepsilon_0,b_0$.
\begin{proof}
The total number of compressed bits probed is equal to
\begin{align*}
  \rwc(\ell_d) &= \sum_{\ell=0}^{\ell_d}\frac{n_{\ell_d}}{n_\ell}(\kla+\klb) \\
               &\leq n_{\ell_d}(R_1+R_2+\gamma_2\varepsilon_0)\\
               &=b_0^{\ell_d+1}4^{\ell_d(\ell_d+1)}(R_1+R_2+\gamma_1\varepsilon_0)
\end{align*}
where $\gamma_1$ is a constant that only depends on $p_{XY},R_1,R_2$.
\end{proof}

\begin{lemma}\label{lemma_ach_pe}
  For any given parameters $(i,\ell_d)$, the probability of error of decoding $X^{\nld},Y^{\nld}$ after decoding up to level $\ell_d$ is upper bounded as follows
  \[
    \pe{\ell_d} \leq 2^{-\beta (\varepsilon_0b_0)^{\ell+1}2^{\ell^2}}=2^{-2^{O(\log \rwc(\ell_d))}}
  \]
\end{lemma}
\begin{proof}
  We will derive the bound by obtaining an upper bound on $\pe{\ell}$ in terms of $\pe{\ell-1}$.
  For  $\ell=0$, we know that
  \[
    \pe{0} \leq 2^{-\beta \varepsilon_0b_0}
  \]
  for a suitable constant $\beta>0$.

  For decoding at level $\ell>1$, there are two possible error events:
  \begin{enumerate}
  \item Event $\cE_1$: More than $\vel b_\ell$ blocks were decoded incorrectly at level $\ell-1$
  \item Event $\cE_2$: There is an incorrect pair of sequences $(\widetilde{x}^{n_{\ell}},\tilde{y}^{n_{\ell}})$ that has the same level-$\ell$ hash/codeword as the true sequence and matches the $(\ell-1)$-level decoded sequence on at least $(1-\vel)b_\ell$ blocks.
  \end{enumerate}
  The overall probability of error is then
  \[
    \pe{\ell} \leq \Pr[\cE_1]+\Pr[\cE_2|\cE_1^c].
  \]
  We will bound the two terms separately. For the first term, observe that
  \begin{align*}
    \Pr[\cE_1] &\leq \nchoosek{b_\ell}{\vel b_{\ell}} \left( \pe{\ell-1} \right)^{\vel b_\ell} &\\
               &\leq \left( \frac{e \pe{\ell-1}}{\vel} \right)^{\vel b_\ell}\\
               &\leq \left( \frac{e 2^{-\beta (\varepsilon_0 b_0)^{\ell}2^{(\ell-1)^2}}}{\varepsilon_0/2^\ell} \right)^{\varepsilon_0 b_0 8^\ell}
  \end{align*}
  where in the last step, we have assumed that $\pe{\ell-1}\leq 2^{-\beta (\varepsilon_0 b_0)^{\ell}2^{(\ell-1)^2}}$. Rewriting the right-hand side, we get
  \begin{align*}
    \Pr[\cE_1] &\leq \exp_2\Bigg( -\beta (\varepsilon_0b_0)^{\ell+1}2^{\ell^2+\ell+1} \\
    &\qquad\qquad+\epsilon_0b_08^\ell\log\left( \frac{e2^\ell}{\varepsilon_0} \right) \Bigg)
  \end{align*}
  Since $\varepsilon_0b_0$ is large enough, the absolute value of the first term in the exponent is at least twice that of the second. Therefore,
  \begin{align}
    \Pr[\cE_1] &\leq \exp_2\left( -\beta (\varepsilon_0b_0)^{\ell+1}2^{\ell^2+\ell}\right) \leq \frac{2^{ -\beta (\varepsilon_0b_0)^{\ell+1}2^{\ell^2}}}{2} &\label{eq:pr_E1}
  \end{align}

  To compute the probability of the second error event, let us define $\cE_2^A$ (resp.\ $\cE_2^B$) to be  the event that there is an incorrect sequences $\widetilde{x}^{n_{\ell}}$ (resp.\ $\tilde{y}^{n_{\ell}})$) that has the same hash as the true sequence and matches the $(\ell-1)$-level decoded sequence on at least $(1-\vel)b_\ell$ blocks.
  We have,
  \begin{align*}
    \Pr[\cE_2^A\vert \cE_1^c] &\leq \nchoosek{b_\ell}{\vel b_{\ell}} |\cX|^{\vel n_{\ell}}2^{-\kla}\\
                              &\leq \left(\frac{e}{\vel}\right)^{\vel b_{\ell}}|\cX|^{\vel n_{\ell}}2^{-\kla}                     
  \end{align*}
  Substituting for $\kla$ in the above and simplifying, we get
  \begin{equation}\label{eq:pr_E2A}
    \Pr[\cE_2^A\vert \cE_1^c] \leq 2^{-\beta \epsilon_\ell n_\ell} \leq \frac{2^{ -\beta (\varepsilon_0b_0)^{\ell+1}2^{\ell^2}}}{4}
  \end{equation}
  Similarly,
  \begin{equation}\label{eq:pr_E2B}
    \Pr[\cE_2^B\vert \cE_1^c] \leq 2^{-\beta \epsilon_\ell n_\ell} \leq \frac{2^{ -\beta (\varepsilon_0b_0)^{\ell+1}2^{\ell^2}}}{4}
  \end{equation}
  Combining \eqref{eq:pr_E1}, \eqref{eq:pr_E2A} and \eqref{eq:pr_E2B}, we get
  \begin{align*}
    \pe{\ell} &\leq \Pr[\cE_1] + \Pr[\cE_2^A\vert \cE_1^c] + \Pr[\cE_2^B\vert \cE_1^c] \\
              &\leq 2^{ -\beta (\varepsilon_0b_0)^{\ell+1}2^{\ell^2}},
  \end{align*}
  which completes the proof.
\end{proof}

\section{Extension to $k>2$ sources}\label{estensione}
We first extend Theorem~\ref{thm:converse_general} to a $k$-source distribution  $p_{X_1,\ldots,X_k}$ defined over alphabet $\cX_1\times\cdots\times \cX_k$. Let the source be $\cX_1$-confusable if for every $x_1,x_1'\in\cX_1$, there exist $(x_2,x_3,\ldots,x_k)$ for which $p_{X_1,\ldots,X_k}(x_1,\ldots,x_k)>0$ and $p_{X_1,\ldots,X_k}(x_1',\ldots,x_k)>0$. Observe that this condition holds if and only if for every $x_1,x_1'\in\cX_1$, there exist an index $i\geq 2$ and $x_i\in\cX_i$ for which $p_{X_1X_i}(x_1,x_i)>0$ and $p_{X_1,X_i}(x_1',x_i)>0$. Now if we repeat the same line of arguments as for the proof of Theorem~\ref{thm:converse_general}, but with the side information $Y^n$ replaced by all sources except $X_1^n$, that is $X_2^n,\ldots, X_k^n$, we get:
\begin{theorem}[Confusable, $k\geq 2$ sources]\label{ext1}
Suppose source $ p_{X_1,\ldots,X_k}$ is $\cX_1$-confusable. If $X^n_1$ is compressed at rate $R_1<H(X_1)$, then
$$\max_{1\leq i\leq n}\Pr[\hat{X}_{1i}(C_{{\cI}_{1i}}(X^n_1),X_2^n,\ldots, X_k^n)\neq X_{1i}]\geq 2^{-\Theta(\rwc)},$$
where $\hat{X}_{1i}(C_{{\cI}_{1i}}(X^n_1),X_2^n,\ldots, X_k^n)$ is any estimator of the $i$-th symbol of source $X_1^n$ given at most $\rwc$ components $C_{{\cI}_{1i}}(X^n_1)$ of $C^{nR_1}(X_1^n)$ and $(X_2^n,\ldots, X_k^n)$.
\end{theorem}

Similarly, Theorem~\ref{lemma:yconfusable_necessary} immediately generalizes to
\begin{theorem}[Non-confusable, $k\geq 2$ sources ]\label{ext2}
  Suppose source $ p_{X_1,\ldots,X_k}$ is not $\cX_1$-confusable.
  Then, it is possible to achieve strong locality at some  $R_1<H(X_1)$ and $R_i=H(X_i)$, $i\in \{2,\ldots,k\}$.
\end{theorem}

The coding scheme of Section~\ref{sec:achievability} easily extends to more than two sources, with the same encoding scheme for each source, and an identical local decoder:  
\begin{theorem}[Hierarchical coding scheme, $k\geq 2$ sources]\label{ext3}
  For any $(R_1,R_2,\ldots,R_k)$ in the interior of the Slepian-Wolf rate region, there exists a rate $(R_1,R_2,\ldots,R_k)$  distributed compression scheme such that for every $1>\eta>2^{-2^{O(\log n)}}$, the local decoder achieves $\rwc=\mathrm{poly}(\log(1/\eta))$ and $\peld\leq \eta$.
\end{theorem}

\section{Concluding Remarks}\label{remfin}
In contrast with the single source set up, we showed that for multiple sources lossless compression and strong locality can generally not be accommodated. For the broad class of confusable sources, for strong locality to hold all sources must be compressed at rates above their respective entropies. On the other hand, if the distribution is not confusable, an arguably peculiar situation, strong locality may hold even if compression rates are below individual entropies. For this case, the characterization of all rate pairs for which strong locality can be achieved remains an open problem.

Our compression scheme is able to achieve $\rwc = {\mathrm{poly}}(\log(1/\eta))$ for any target probability of local decoding error $\eta$ specified at the decoder. Note that from our lower bound, $\rwc =\Omega(\log(1/\eta))$ and our scheme is suboptimal by a polynomial factor. Designing an improved scheme that achieves this lower bound is left as future work.

In this paper, we only considered the problem of local decodability in the context of distributed compression. One may also require provisioning of local substitutions/insertions/deletions of source symbols in the compressed domain. This is an interesting problem that warrants more attention.
\bibliographystyle{IEEEtran}
\bibliography{locality_references,codes_with_locality}

\begin{thebibliography}{10}
\providecommand{\url}[1]{#1}
\csname url@samestyle\endcsname
\providecommand{\newblock}{\relax}
\providecommand{\bibinfo}[2]{#2}
\providecommand{\BIBentrySTDinterwordspacing}{\spaceskip=0pt\relax}
\providecommand{\BIBentryALTinterwordstretchfactor}{4}
\providecommand{\BIBentryALTinterwordspacing}{\spaceskip=\fontdimen2\font plus
\BIBentryALTinterwordstretchfactor\fontdimen3\font minus
  \fontdimen4\font\relax}
\providecommand{\BIBforeignlanguage}[2]{{%
\expandafter\ifx\csname l@#1\endcsname\relax
\typeout{** WARNING: IEEEtran.bst: No hyphenation pattern has been}%
\typeout{** loaded for the language `#1'. Using the pattern for}%
\typeout{** the default language instead.}%
\else
\language=\csname l@#1\endcsname
\fi
#2}}
\providecommand{\BIBdecl}{\relax}
\BIBdecl

\bibitem{vatedka2022_conferenceversion}
S.~Vatedka, V.~Chandar, and A.~Tchamkerten, ``Locally decodable {Slepian-Wolf}
  compression,'' in \emph{2022 IEEE International Symposium on Information
  Theory (ISIT)}.\hskip 1em plus 0.5em minus 0.4em\relax IEEE, 2022, pp.
  1430--1435.

\bibitem{pavlichin2013human}
D.~S. Pavlichin, T.~Weissman, and G.~Yona, ``The human genome contracts
  again,'' \emph{Bioinformatics}, vol.~29, no.~17, pp. 2199--2202, 2013.

\bibitem{pavlichin2018quest}
D.~Pavlichin and T.~Weissman, ``The quest to save genomics: Unless researchers
  solve the looming data compression problem, biomedical science could
  stagnate,'' \emph{IEEE Spectrum}, vol.~55, no.~9, pp. 27--31, 2018.

\bibitem{chen2014data}
C.~P. Chen and C.-Y. Zhang, ``Data-intensive applications, challenges,
  techniques and technologies: A survey on big data,'' \emph{Information
  Sciences}, vol. 275, pp. 314--347, 2014.

\bibitem{hashem2015rise}
I.~A.~T. Hashem, I.~Yaqoob, N.~B. Anuar, S.~Mokhtar, A.~Gani, and S.~U. Khan,
  ``The rise of “big data” on cloud computing: Review and open research
  issues,'' \emph{Information systems}, vol.~47, pp. 98--115, 2015.

\bibitem{ball2012public}
M.~P. Ball, J.~V. Thakuria, A.~W. Zaranek, T.~Clegg, A.~M. Rosenbaum, X.~Wu,
  M.~Angrist, J.~Bhak, J.~Bobe, M.~J. Callow \emph{et~al.}, ``A public resource
  facilitating clinical use of genomes,'' \emph{Proceedings of the National
  Academy of Sciences}, vol. 109, no.~30, pp. 11\,920--11\,927, 2012.

\bibitem{uk10k2015uk10k}
U.~consortium \emph{et~al.}, ``The uk10k project identifies rare variants in
  health and disease,'' \emph{Nature}, vol. 526, no. 7571, p.~82, 2015.

\bibitem{schadt2010computational}
E.~E. Schadt, M.~D. Linderman, J.~Sorenson, L.~Lee, and G.~P. Nolan,
  ``Computational solutions to large-scale data management and analysis,''
  \emph{Nature reviews genetics}, vol.~11, no.~9, p. 647, 2010.

\bibitem{ziv1977universal}
J.~Ziv and A.~Lempel, ``A universal algorithm for sequential data
  compression,'' \emph{IEEE Transactions on Information Theory}, vol.~23,
  no.~3, pp. 337--343, 1977.

\bibitem{ziv1978compression}
------, ``Compression of individual sequences via variable-rate coding,''
  \emph{IEEE Transactions on Information Theory}, vol.~24, no.~5, pp. 530--536,
  1978.

\bibitem{mazumdar2015local}
A.~Mazumdar, V.~Chandar, and G.~W. Wornell, ``Local recovery in data
  compression for general sources,'' in \emph{Proceedings of the 2015 IEEE
  International Symposium on Information Theory (ISIT)}.\hskip 1em plus 0.5em
  minus 0.4em\relax IEEE, 2015, pp. 2984--2988.

\bibitem{tatwawadi18isit_universalRA}
K.~Tatwawadi, S.~Bidokhti, and T.~Weissman, ``On universal compression with
  constant random access,'' in \emph{Proceedings of the 2018 IEEE International
  Symposium on Information Theory}, 2018, pp. 891--895.

\bibitem{patrascu2008succincter}
M.~Patrascu, ``Succincter,'' in \emph{2008 49th Annual IEEE Symposium on
  Foundations of Computer Science}.\hskip 1em plus 0.5em minus 0.4em\relax
  IEEE, 2008, pp. 305--313.

\bibitem{dodis2010changing}
Y.~Dodis, M.~Patrascu, and M.~Thorup, ``Changing base without losing space,''
  in \emph{Proceedings of the forty-second ACM symposium on Theory of
  computing}, 2010, pp. 593--602.

\bibitem{munro2015compressed}
J.~I. Munro and Y.~Nekrich, ``Compressed data structures for dynamic
  sequences,'' in \emph{Algorithms-ESA 2015}.\hskip 1em plus 0.5em minus
  0.4em\relax Springer, 2015, pp. 891--902.

\bibitem{raman2003succinct}
R.~Raman and S.~S. Rao, ``Succinct dynamic dictionaries and trees,'' in
  \emph{International Colloquium on Automata, Languages, and
  Programming}.\hskip 1em plus 0.5em minus 0.4em\relax Springer, 2003, pp.
  357--368.

\bibitem{chandar2009locally}
V.~Chandar, D.~Shah, and G.~W. Wornell, ``A locally encodable and decodable
  compressed data structure,'' in \emph{Proceedings of the 47th Annual Allerton
  Conference on Communication, Control, and Computing}.\hskip 1em plus 0.5em
  minus 0.4em\relax IEEE, 2009, pp. 613--619.

\bibitem{chandar_thesis}
V.~B. Chandar, ``Sparse graph codes for compression, sensing and secrecy,''
  Ph.D. dissertation, MIT, 2010.

\bibitem{kreft2010lz77}
S.~Kreft and G.~Navarro, ``{LZ77}-like compression with fast random access,''
  in \emph{2010 Data Compression Conference}.\hskip 1em plus 0.5em minus
  0.4em\relax IEEE, 2010, pp. 239--248.

\bibitem{dutta2013simple}
A.~Dutta, R.~Levi, D.~Ron, and R.~Rubinfeld, ``A simple online competitive
  adaptation of lempel-ziv compression with efficient random access support,''
  in \emph{Proceedings of the Data Compression Conference (DCC)}.\hskip 1em
  plus 0.5em minus 0.4em\relax IEEE, 2013, pp. 113--122.

\bibitem{bille2011random}
P.~Bille, G.~M. Landau, R.~Raman, K.~Sadakane, S.~R. Satti, and O.~Weimann,
  ``Random access to grammar-compressed strings,'' in \emph{Proceedings of the
  twenty-second annual ACM-SIAM symposium on Discrete Algorithms}.\hskip 1em
  plus 0.5em minus 0.4em\relax Society for Industrial and Applied Mathematics,
  2011, pp. 373--389.

\bibitem{viola2019howtostore}
E.~Viola, O.~Weinstein, and H.~Yu, ``How to store a random walk,'' \emph{arXiv
  preprint arXiv:1907.1087}, 2019.

\bibitem{sadakane2006squeezing}
K.~Sadakane and R.~Grossi, ``Squeezing succinct data structures into entropy
  bounds,'' in \emph{Proceedings of the seventeenth annual ACM-SIAM symposium
  on Discrete algorithm}.\hskip 1em plus 0.5em minus 0.4em\relax Society for
  Industrial and Applied Mathematics, 2006, pp. 1230--1239.

\bibitem{gonzalez2006statistical}
R.~Gonz{\'a}lez and G.~Navarro, ``Statistical encoding of succinct data
  structures,'' in \emph{Annual Symposium on Combinatorial Pattern
  Matching}.\hskip 1em plus 0.5em minus 0.4em\relax Springer, 2006, pp.
  294--305.

\bibitem{ferragina2007simple}
P.~Ferragina and R.~Venturini, ``A simple storage scheme for strings achieving
  entropy bounds,'' in \emph{Proceedings of the eighteenth annual ACM-SIAM
  symposium on Discrete algorithms}.\hskip 1em plus 0.5em minus 0.4em\relax
  Society for Industrial and Applied Mathematics, 2007, pp. 690--696.

\bibitem{makhdoumi_locally_2015}
A.~Makhdoumi, S.-L. Huang, M.~Medard, and Y.~Polyanskiy, ``On locally decodable
  source coding,'' in \emph{2015 {IEEE} {International} {Conference} on
  {Communications} ({ICC})}, London, Jun. 2015, pp. 4394--4399.

\bibitem{nicholson2013survey}
P.~K. Nicholson, V.~Raman, and S.~S. Rao, ``A survey of data structures in the
  bitprobe model,'' in \emph{Space-Efficient Data Structures, Streams, and
  Algorithms}.\hskip 1em plus 0.5em minus 0.4em\relax Springer, 2013, pp.
  303--318.

\bibitem{makhdoumi2013locally-arxiv}
A.~Makhdoumi, S.-L. Huang, M.~Medard, and Y.~Polyanskiy, ``On locally decodable
  source coding,'' \emph{arXiv preprint arXiv:1308.5239}, 2013.

\bibitem{buhrman2002bitvectors}
H.~Buhrman, P.~B. Miltersen, J.~Radhakrishnan, and S.~Venkatesh, ``Are
  bitvectors optimal?'' \emph{SIAM Journal on Computing}, vol.~31, no.~6, pp.
  1723--1744, 2002.

\bibitem{pananjady2018effect}
A.~Pananjady and T.~A. Courtade, ``The effect of local decodability constraints
  on variable-length compression,'' \emph{IEEE Transactions on Information
  Theory}, vol.~64, no.~4, pp. 2593--2608, 2018.

\bibitem{vatedka_local_2020}
S.~Vatedka and A.~Tchamkerten, ``Local decode and update for big data
  compression,'' \emph{IEEE Transactions on Information Theory}, vol.~66,
  no.~9, pp. 5790--5805, 2020.

\bibitem{vatedka2020log}
S.~Vatedka, V.~Chandar, and A.~Tchamkerten, ``O (log log n) worst-case local
  decoding and update efficiency for data compression,'' in \emph{2020 IEEE
  International Symposium on Information Theory (ISIT)}, Los Angeles, CA, USA,
  2020, pp. 2371--2376.

\bibitem{vestergaard2021enabling}
R.~Vestergaard, Q.~Zhang, and D.~E. Lucani, ``Enabling random access in
  universal compressors,'' in \emph{IEEE INFOCOM 2021-IEEE Conference on
  Computer Communications Workshops (INFOCOM WKSHPS)}.\hskip 1em plus 0.5em
  minus 0.4em\relax IEEE, 2021, pp. 1--6.

\bibitem{kamparaju2022low}
S.~Kamparaju, S.~Mastan, and S.~Vatedka, ``Low-complexity compression with
  random access,'' in \emph{2022 IEEE International Conference on Signal
  Processing and Communications (SPCOM)}.\hskip 1em plus 0.5em minus
  0.4em\relax IEEE, 2022, pp. 1--5.

\bibitem{slepian1973noiseless}
D.~Slepian and J.~Wolf, ``Noiseless coding of correlated information sources,''
  \emph{IEEE Transactions on Information Theory}, vol.~19, no.~4, pp. 471--480,
  1973.

\bibitem{csiszar1980towards}
I.~Csisz\'ar and J.~K\"orner, ``Towards a general theory of source networks,''
  \emph{IEEE Transactions on Information Theory}, vol.~26, no.~2, pp. 155--165,
  1980.

\bibitem{kamath_reverse_2015}
S.~Kamath, ``Reverse hypercontractivity using information measures,'' in
  \emph{2015 53rd {Annual} {Allerton} {Conference} on {Communication},
  {Control}, and {Computing} ({Allerton})}, Sep. 2015, pp. 627--633.

\end{thebibliography}
	
\end{document}